\documentclass{article}

\usepackage{graphicx} 
\def\withcolors{0}
\def\withnotes{0}
\usepackage{ccanonne}
\usepackage{multicol}
\usepackage{multirow}
\usepackage{subcaption}
\usepackage{bbm}
\usepackage{braket}
\usepackage{algorithm}
\usepackage{algpseudocode}
\usepackage{turnstile}

\newcommand{\tr}{\operatorname{Tr}}

\DeclareMathOperator*{\argmin}{arg\,min}

\allowdisplaybreaks
\newcolumntype{D}{>{\centering\arraybackslash}p{0.26\textwidth}}
\newcolumntype{C}{>{\centering\arraybackslash}p{0.23\textwidth}}
\newcommand{\dims}{d}

\newcommand{\nspu}{m}

\newcommand{\cD}{{\mathcal{D}}}

\newcommand{\povmset}{{\mathfrak{M}}}
\newcommand{\nqubits}{{N}}

\newcommand{\opnorm}[1]{{\left\|#1\right\|}_{\text{op}}}
\newcommand{\tracenorm}[1]{{\left\|#1\right\|}_{1}}
\newcommand{\hsnorm}[1]{{\left\|#1\right\|}_{\text{HS}}}

\newcommand{\supparen}[1]{^{(#1)}}

\newcommand{\cA}{\mathcal{A}}

\newcommand{\bfP}{\mathbf{P}}

\newcommand{\x}{\mathbf{x}}
\newcommand{\out}{{x}}

\def\multiset#1#2{\ensuremath{\left(\kern-.3em\left(\genfrac{}{}{0pt}{}{#1}{#2}\right)\kern-.3em\right)}}

\newcommand{\ham}[2]{\operatorname{d}_{\rm Ham}\left(#1,#2\right)}
\newcommand{\EMD}[2]{\operatorname{d}_{\rm EM}\left(#1,#2\right)}

\newcommand*{\proj}[1]{|#1\rangle\!\langle #1|}

\newcommand{\qmm}{{\rho_{\text{mm}}}} 

\newcommand{\HH}{\mathbb{H}}
\newcommand{\Herm}[1]{{\HH_{#1}}}

\newcommand{\qbit}[1]{|{#1}\rangle}
\newcommand{\qadjoint}[1]{\langle{#1}|}
\newcommand{\qproj}[1]{\qbit{#1}\qadjoint{#1}}

\newcommand{\qdotprod}[2]{\langle#1|#2\rangle}
\newcommand{\hdotprod}[2]{\left\langle#1,#2\right\rangle}
\newcommand{\matdotprod}[3]{\langle#1|#2|#3\rangle}

\newcommand{\eye}{\mathbb{I}}

\newcommand{\rk}{{r}}

\newcommand{\bx}{\mathbf{x}}
\newcommand{\outset}{{\mathcal{X}}}

\newcommand{\dm}{\mathrm{d}}
\newcommand{\Sp}{\mathbb{S}}

\newcommand{\Sim}{\mathcal{S}}

\newcommand{\Haar}[1]{{\mathcal{U}_{#1}}}
\newcommand{\POVM}{\mathcal{M}}

\newcommand{\cycle}{\mathcal{C}}

\newcommand{\nobs}{M}
\newcommand{\Obs}{O}

\title{Shadow Tomography Against Adversaries}

\author{
    \begin{tabular}[t]{c c c c}
   Maryam Aliakbarpour\thanks{Computer science department and Ken Kennedy Institute.} &Vladimir Braverman & Nai-Hui Chia &  Chia-Ying Lin \\
 Rice University & Johns Hopkins University & Rice University & Rice University\\ 
\small \texttt{maryama@rice.edu} &\small \texttt{vova@cs.jhu.edu} & \small \texttt{nc67@rice.edu} & \small \texttt{cl207@rice.edu}
\end{tabular}
\vspace{2ex}\\
\begin{tabular}[t]{c c c}
Yuhan Liu & Aadil Oufkir & Yu-Ching Shen\\
Rice University & University Mohammed VI Polytechnic & Rice University\\ 
\small \texttt{yuhan-liu@rice.edu} &\small\texttt{aadil.oufkir@gmail.com} & \small\texttt{ycshen@rice.edu}
\end{tabular}
}

\ifdefined\authorNote
\newcommand{\yc}[1]{{\color{magenta} [{YC:} #1]}}
\newcommand{\nai}[1]{{\color{brown} [{Nai:} #1]}}

\else
\newcommand{\yc}[1]{}
\newcommand{\nai}[1]{}
\fi
\newcommand{\remove}{\xi}
\newcommand{\truncation}{\xi}
\newcommand{\corruption}{\gamma}

\begin{document}
\maketitle
\begin{abstract}
Learning about quantum states is a fundamental problem in physics and quantum computing. As people are often interested in certain properties of quantum states instead of a complete description, shadow tomography has gained significant attention, where the goal is to learn the expectation values of $M$ observables $O_1, \ldots, O_M$  with $\varepsilon$ accuracy. 
In near-term devices, however, noise is prevalent and often unexpected. Thus, it is crucial to design algorithms that work well in the worst case. 

We study the practical single-copy setting and assume $\gamma$-fraction of the \emph{outcomes} can be arbitrarily corrupted by an adversary.
We show that all non-adaptive shadow tomography algorithms must incur an error of $\varepsilon=\tilde{\Omega}(\gamma\min\{\sqrt{M}, \sqrt{d}\})$ for some choice of observables, even with unlimited copies. 
Unfortunately, the classical shadows algorithm by \cite{huang2020predicting} and naive algorithms that directly measure each observable suffer even more.
We design an algorithm that achieves an error of $\eps=\tilde{O}(\gamma\max_{i\in[M]}\|O_i\|_{HS})$, which nearly matches our worst-case error lower bound for $M\ge d$ and guarantees better accuracy when the observables have stronger structure. 
Remarkably, the algorithm only needs $n=\frac{1}{\gamma^2}\log(M/\delta)$ copies to achieve that error with probability at least $1-\delta$, matching the sample complexity of the classical shadows algorithm that achieves the same error without corrupted measurement outcomes. Our algorithm is conceptually simple and easy to implement. Classical simulation for fidelity estimation shows that our algorithm enjoys much stronger robustness than~\cite{huang2020predicting} under adversarial noise.

Finally, based on a reduction from full-state tomography to shadow tomography, we prove that for rank $r$ states, both the near-optimal asymptotic error of $\eps=\tilde{O}(\gamma\sqrt{r})$ \emph{and} copy complexity $\tilde{O}(dr^2/\eps^2)=\tilde{O}(dr/\gamma^2)$ can be achieved for adversarially robust state tomography, closing the large gap in \cite{AliakbarpourBCL2025robustquantum} where optimal error can only be achieved using pseudo-polynomial number of copies in $\dims$.

\end{abstract}

\ynote{
\begin{enumerate}
    \item HKP lower bound (section 4.3 , Chia-Ying, Yu-Ching)
    \item Comments in the algorithm proof
    \item Go over introduction (Nai)
    \item Preliminary paraphrase

\end{enumerate}
}

\section{Introduction}
Learning about quantum states is a fundamental problem in physics and quantum computing.
In the canonical setting, we are given identical copies of an unknown state $\rho$ and need to design measurements and algorithms to learn information about $\rho$. 
While learning the full state description is desirable, it is often prohibitively costly. Even in the most ideal scenario where we can store all copies and perform fully quantum operations, the number of copies needed scales as $\Theta(4^{\nqubits})$ for $\nqubits$-qubit states in the worst case \cite{ODonnellW16,HaahHJWY17}. 
The scaling is even worse when we measure each copy or each qubit separately, which are more practical measurement setups in the near-term: $\Theta(8^{\nqubits})$~\cite{guctua2020fast, HaahHJWY17} for the former and $\tildeTheta{10^{\nqubits}}$\footnote{Here $\tilde{O}$ hides a $\sqrt{\nqubits}$ factor.}~\cite{ADLY2025Paulinot,ADLY2025PauliOpt} for the latter.

In practice, people are often interested in some important properties of the state that they care about, and thus learning the full state could be a significant overshoot. This gives rise to shadow tomography~\cite{Aaronson20}, where the goal is to simultaneously learn the expectation values of a finite set of $\nobs$ observables. Remarkably, there exist algorithms for both fully entangled (e.g.,~\cite{Aaronson2018,Aaronson20,buadescu2021improved}) and single-copy measurements (e.g.,~\cite{huang2020predicting}) where the number of copies scales only logarithmically with the number of observables.

However, a major challenge for quantum computing is the presence of noise, which makes the learning task more difficult. 
Despite recent progress~\cite{google2023suppressing,acharya2024quantum} in error correction that improves the quality of quantum hardware, in principle, noise cannot be fully eliminated. 
Moreover, many non-physical factors, such as adversarial human factors, cannot be dealt with by error correction. 
When quantum computing improves, human adversaries would also possess more power to sabotage learning systems.
Thus, improving the robustness of quantum algorithms is still extremely important in the present and in the future.

Many works~\cite{yu2023robust,stilck2024efficient, brandao2020fast} have studied robustness to noise for quantum state learning, including various works on robust shadow tomography~\cite{ChenYuZengFlammia2020robustshadow,hu2025demonstration,Brieger2025stability}. 
However, they often require structural noise assumptions, such as Markovian, time-independence, and other geometric assumptions on the underlying noise channel. 
The complication of real-world physical environment could easily break these assumptions and make their algorithms invalid. 
For example, if the cooling system of superconducting quantum circuit becomes unstable for a small time window during the experiment, the noise that occurs during this time would break the time-independence assumption, and likely others as well.

To account for the unknown and unexpected nature of noise in the real-world, a recent work of~\cite{AliakbarpourBCL2025robustquantum} initiated the study of adversarial corruption model for full-state tomography in the single-copy measurement setting, where they assume that $\gamma$-fraction of the measurement outcomes can be arbitrarily and even adversarially corrupted. 
They prove that an error of $\tildeTheta{\gamma\sqrt{\rk}}$ in trace distance is necessary and achievable with sufficient number of copies to learn a rank-$\rk$ state.
One can indeed directly take their algorithm and apply it to shadow tomography, but it will be a significant overshoot, not only because shadow tomography is a easier problem, but also because their algorithm is extremely inefficient even for full-state tomography: the number of copies required is $O(\dims^{O(\log(1/\gamma))})$, not even polynomial in the dimensionality $\dims=2^\nqubits$.

Due to the significance of shadow tomography in practice and the prevalence of noise in real-world quantum systems, it is paramount to design efficient and robust algorithms for the problem. 
Surprisingly, an extremely simple modification to the classical shadows algorithm~\cite{huang2020predicting} not only provides provably optimal robustness guarantees against worst-case adversaries, but also makes no harm to the sample/copy efficiency. 
The theoretical results are supported by simulation experiments on classical computers. 
As a by-product, we nearly resolve the optimal sample complexity for robust full-state tomography via a reduction to shadow tomography.

\subsection{Adversarial corruption model for single-copy shadow tomography}
We begin with the formal problem formulation before introducing our results and techniques.
There are $\ns$ copies of an unknown $\nqubits$-state with density matrix $\rho\in \C^{\dims\times\dims}, \dims=2^{\nqubits}$ and a random seed $R$, which can be viewed as an infinite random binary string. We can apply measurements $\POVM^\ns = (\POVM_1, \ldots, \POVM_\ns)$ to each copy, where $\POVM_i=\{M_x^i\}_{x\in\cX}$. For $i\ge 1$ let $\out_i$ be the outcome of measuring the $i$th copy with $\POVM_i$. All measurements can be chosen based on the random seed $R$. Define $\out^t=(\out_0,\out_1, \ldots, \out_t)$. 

\begin{description}
\item[Shadow tomography.] Given an unknown state and $M$ observables $O_1, \ldots, O_M$ with $\opnorm{O_i}\le 1$, the goal is to design a measurement scheme $\POVM^\ns$ and output estimates $E_1, \ldots, E_M$ such that with probability at least $1-\delta$,
\[
\forall i\in[M], |E_i-\Tr[O_i\rho]|\le \eps.
\]
The probability is over the random seed $R$ and randomness in the measurement outcomes.

\end{description}

When the state is $\rho$, the distribution of $\out_i,i\ge 1$ is determined by Born's rule,
\begin{equation}
    \p_\rho^{i}(x)=\Tr[M_x^i\rho],\label{equ:distr-i-cond}
\end{equation}
For $1\le t\le \ns$, we further define $\p_\rho^{\out^t}$ as the distribution of $\out^t$ when the state is $\rho$. For non-adaptive measurements, $\p_\rho^{\out^t}$ is a product distribution conditioned on the random seed $R$. The measurements $\POVM^\ns$ can also be chosen adaptively, i.e., $\POVM_t$ depends on previous outcomes $x^{t-1}$. In this case, $\p_\rho^{\out^t}$ would not be a product distribution in general.

In practice, we may have restrictions on the types of measurements that can be applied. We use $\povmset$ to denote the set of allowable measurements for each copy. 

\paragraph{Adversarial corruption.} The noise model is defined as follows,
\begin{definition}[$\gamma$-adversarial corruption model]
    For non-adaptive measurement schemes, the interaction between the adversary and the algorithm proceeds as,
\begin{enumerate}
    \item Measurements are applied to all copies of $\rho$ to obtain outcomes $(x_1, \ldots, x_\ns)$.
    \item The adversary $\cA$ arbitrarily changes a $\gamma$-fraction of them. It has perfect knowledge about the measurements used by the algorithm and can perform the corruption based on that knowledge. 
    \item The algorithm uses the corrupted outcomes $(y_1, \ldots, y_\ns)$ to learn about the state $\rho$.
\end{enumerate}
\end{definition}

\subsection{Results}

We first prove a lower bound that tells the best achievable error even when the number of copies is infinite. Similar to robust full state tomography~\cite{AliakbarpourBCL2025robustquantum}, a dimension-dependent error is inevitable in the worst case under advsersarial corruption.
\begin{theorem}
\label{thm:result-lower}
    Under the $\gamma$-adversarial corruption model, there exists a set of observables $\Obs_1, \ldots, \Obs_{\nobs}$  such that all non-adaptive single-copy shadow tomography algorithms must incur an error of
    \[
    \eps = \tildeOmega{\gamma\min\{\sqrt{\nobs}, \sqrt{\dims}\}}.
    \]
\end{theorem}
\begin{remark}
    By restricting the unknown state and observables to the same $r$-dimensional subspace, we can obtain a lower bound of $\eps = \tildeOmega{\gamma\min\{\sqrt{\nobs}, \sqrt{r}}.$
\end{remark}

We further design a robust shadow tomography algorithm that achieves the optimal error when $\nobs$ is large. Moreover, the copy complexity matches the standard algorithm by~\cite{huang2020predicting} that achieves the same error without adversarial corruption.
\begin{theorem}
\label{thm:result-algorithm}
    Under $\gamma$-adversarial corruption, there exists a shadow tomography algorithm such that for all observables $\Obs_1, \ldots, \Obs_{\nobs}$, achieves an error of at most $\eps=\bigO{\gamma{\max_{i\in[\nobs]}\hsnorm{O_i}}}$ with probability at least $1-\delta$. The number of copies used by the algorithm is at most
    \[
    \ns = \bigO{\frac{\max_{i\in[\nobs]}\hsnorm{O_i}^2}{\eps^2}\log\frac{\nobs}{\delta}}=\bigO{\frac{1}{\gamma^2}\log\frac{\nobs}{\delta}}.
    \]
\end{theorem}

By a reduction from full state tomography to shadow tomography, our shadow tomography algorithm leads to a sample optimal robust tomography algorithm. 
\begin{corollary}
    Under $\gamma$-adversarial corruption, there exists an full state tomography algorithm that achieves an error of $\eps = \tildeO{\gamma\sqrt{\rk}}$ in trace distance using $\ns= \tildeO{\dims\rk/\gamma^2}$ samples.
\end{corollary}
Our algorithm achieves the min-max optimal error lower bound of $\Omega(\gamma\sqrt{\rk})$ in~\cite{AliakbarpourBCL2025robustquantum} and the optimal sample complexity up to log factors: for $\eps=\gamma\sqrt
{\rk}$, $\ns=\Omega(\dims\rk^2/\eps^2)=\Omega(\dims\rk/\gamma^2)$ copies are necessary to learn an unknown state with $\eps$ accuracy  even without corruption using non-adaptive measurements~\cite{HaahHJWY17}. 
For $\rk=\dims$, it also matches the lower bound of $\Omega(\dims^3/\eps^2)$ for adaptive measurements~\cite{chen2023does} (note that the exact lower bound for rank-$r$ states in is unknown to date in the adaptive setting)
In contrast, the algorithm in~\cite{AliakbarpourBCL2025robustquantum} required pseudo-polynomial number of copies.

\subsection{Related works}

\paragraph{Shadow tomography} Under the promise that all  copies can be measured coherently, shadow tomography can be performed with a poly-logarithmic complexity in both the system dimension and the number of observables \cite{Aaronson2018,Aaronson20,buadescu2021improved}. For protocols restricted to incoherent measurements, the classical shadows framework \cite{Paini2019Oct,Morris2019Sep,huang2020predicting,Neven2021Oct,Bertoni2024Jul,Fanizza2024Jun,Helsen2023Dec,Akhtar2023Jun,Wan2023Dec,Low2022Aug} provides efficient estimation for specific observable families. However, poly-logarithmic complexity in the dimension is not possible in general \cite{huang2020predicting}. 

\cite{ChenYuZengFlammia2020robustshadow,hu2025demonstration} propose robust shadow tomography protocols in which measurement noise can be mitigated by a calibration procedure that estimates the inverse of the noisy measurement channel. 
Our work considers a different and more challenging noise model, where a fraction of the measurement outcomes are completely corrupted. These corrupted outcomes are not merely a noisy version of the true measurement (e.g., a depolarized result) but may be entirely arbitrary or even adversarially chosen. This distinction also holds under general gate-dependent, time-stationary, and Markovian noise \cite{Brieger2025stability} and more generally the non-iid assumption \cite{Fawzi2024Nov}. Consequently, while the protocols in the aforementioned works can achieve an arbitrary precision $\eps$, our protocol's precision is lower-bounded by an expression that depends on the corruption fraction and other problem parameters.

\paragraph{State tomography} When coherent measurements are permitted, the complexity of state tomography is $ \widetilde{\Theta}(dr/\eps^2)$ \cite{o2016efficient,Haah2016Jun}, while it becomes $\widetilde{\Theta}(dr^2/\eps^2)$ for non-adaptive incoherent, single-copy measurements \cite{Kueng2017Jan,guctua2020fast,Haah2016Jun}; here, $\eps$ represents the approximation error and $r$ is the rank of the unknown state. For algorithms using adaptive incoherent, single-copy measurements, only the lower bound $\Omega(d^3/\eps^2)$ in the worst case  is known  \cite{chen2023does}. These algorithms work in the uncorrupted model. 

\paragraph{Learning with noise}
There is a large body of work on learning quantum states in the presence of noise. Many requires specific statistical or structural assumptions on the noise
\cite{yu2023robust,stilck2024efficient,rambach2021robust,brandao2020fast,endo2018practical, cai2023quantum}.
\cite{Jayakumar2024universalframework} studied learning the state and noise model simultaneously, but due to the inherent ambiguity, they can only be learned up to an ambiguity factor.

Agnostic learning~\cite{chung2018sample,buadescu2021improved} aims to learn a closest state in a specific concept class, such as stabilizer states~\cite{chen2024stabilizer} and product states~\cite{bakshi2024learning, grewal2024agnostic}. One can view this setup as having independent and identical noise that deviates the true state from the ideal concept class.

The first tomographic algorithm in the adversarial corruption model is proposed by \cite{AliakbarpourBCL2025robustquantum} and achieves an accuracy $\eps = O(\gamma \sqrt{r})$ and has a copy complexity $d^{O(\log 1/\gamma)} /\gamma^2$ where $r$ is the state's rank and $\gamma$ is the corruption parameter. In this work, we improve on the copy complexity by proposing a new algorithm requiring only $\widetilde{O}(dr^2/\eps^2)$ copies, which is optimal up to logarithmic factors. 
A later but concurrent work \cite{Arulandu2025Oct} uses robust statistics to design an efficient algorithm to learn  quantum states close to product, which also works under adversarial noise. These two works extend the study of learning quantum states to the worst-case adversarial noise setting.

\paragraph{Classical robust statistics} Robust statistics~\cite{huber1964robust,lugosi2021robust} is an active area of research that deals with statistical problems under noise models. A large and growing literature has made substantial progress on designing time-efficient algorithms for key problems, including mean estimation ~\cite{CDG17,DiakonikolasKP20,7782980,Hopkins2020robustheavy,hopkins2025subGmean}, learning mixtures of Gaussians~\cite{BakshiDHKKK20robustSoS,Liu2023robustmixture}, and covariance estimation~\cite{kothari2018robust,BakshiDHKKK20robustSoS,ilias2024SoSsubG}. In our results, a dimension dependent error is unavoidable in the worst case setting. This phenomenon was observed in a different classical context:  distributed estimation under local differential privacy and communication constraints~\cite{acharya21manipulation}.

\section{Technical overview}

\subsection{Failure of existing algorithms under adversarial noise}\label{sec:tomography-overview}
\paragraph{The naive algorithm} Perhaps the  simplest algorithm for shadow tomography is to split the copies evenly into $\nobs$ groups and directly measure the expectation values of each observable in each group. 
If we normalize the observables such that $0\preceq O_i\preceq \eye_\dims$, this means that we apply the measurement $\{O_i, \eye_\dims-O_i\}$ for each group of $\nspu/\nobs$ copies and approximate $\Tr[\rho O_i]$ by the empirical mean. 
However, it is straightforward to see that this algorithm would fail drastically under adversarial corruption: the adversary only needs to pick one observable, say $O_1$, and flip some $\gamma'$ fraction of the binary outcomes. Our estimate for $O_i$ would shift by $\gamma'$ (and no more than that),  but only $\gamma'/\nobs$ fraction of \emph{all} the samples are corrupted. 
Equivalently, this means that corrupting $\gamma$ fraction of the measurement outcomes can lead to an error of $\gamma \nobs$ in the worst case.

One may speculate that the failure is because empirical mean estimators are generally not robust to noise. However, we show that the $\gamma M$ lower bound holds generally for all strategies that measure each copy with $\{O_i, \eye_\dims-O_i\}$. Thus, this additional error is a fundamental limit and cannot be improved by using other post-processing approaches such as quantile or median-of-means estimators. The lower bound is proved in~\cref{sec:lb-direct} its ideas are explained in~\cref{sec:overview-lb}.
Our discussion leads to the following theorem for algorithms using $\{O_i, \eye_\dims-O_i\}$. 
\begin{theorem}
\label{thm:naive}
    Under $\gamma$-adversarial corruption, the algorithm that directly measures each observable $O_i$ using $\ns/\nobs$ copies and outputs their empirical means achieves an error of $\eps = O(\min\{\gamma \nobs, 1\})$ using $\ns = \nobs\log(\nobs/\delta)/\eps^2$ copies. This is optimal in the sense that all algorithms that directly measure each copy with one of $\{O_i, \eye_\dims-O_i\}$ must incur an error of $\eps=\Omega(\gamma \nobs)$.
\end{theorem}

\paragraph{Classical shadows by~\cite{huang2020predicting}} We review the shadow tomography algorithm via classical shadows in~\cite{huang2020predicting}. We apply the uniform POVM to all copies, 
\[
\POVM_{unif}=\{\dims\qproj{v}\}_{v\in \Sp^{\dims} }.
\]
Each $\qbit{v}$ is drawn uniformly from the Haar measure, i.e., a uniform distribution over the complex unit sphere.
When applying the measurement to $\rho$, by Born's rule, we obtain a vector $\qbit{v}$ as an outcome which occurs with a probability density of 
\begin{equation}
\label{equ:uniform-povm-density}
    \dims\matdotprod{v}{\rho}{v}\dm v.
\end{equation}
We denote this distribution as $\cD(\rho)$. Using \eqref{equ:haar-k-moment}, the expectation of $\qproj{v}$ can be written in terms of the unknown state $\rho$,
\begin{equation}
    \Sigma_\rho \eqdef \expectDistrOf{v\sim\cD(\rho)}{\qproj{v}} =\int\qproj{v}\matdotprod{v}{\rho}{v}\dm v= \frac{\eye_\dims+\rho}{\dims+1}.
    \label{equ:sigma-rho}
\end{equation}
Therefore, the expectation value of an observable $\Obs_i$ can be computed as follows,
\begin{equation}
\label{equ:classical-shadow-mean}
    \Tr[\Obs_i\rho]=(\dims+1)\expectDistrOf{v\sim\cD(\rho)}{\matdotprod{v}{\Obs_i}{v}}-\Tr[\Obs_i].
\end{equation}

By measuring all copies with the uniform POVM, we obtain a set of outcomes $\qbit{v_1}, \ldots, \qbit{v_\ns}$. For each observable $\Obs_i$, we can then approximate the above quantity using the empirical mean of $\matdotprod{v_j}{O_i}{v_j}$.
The algorithm in~\cite{huang2020predicting} makes a slight modification and instead uses median-of-means (MoM): we divide $\ns$ samples into  batches of equal size $K$, compute the mean of each batch, and return the median of the $\ns/K$ empirical means. 
This strategy indeed offers some level of robustness when only a small fraction of batches are corrupted, which is better than the plain empirical mean estimator. 

Unfortunately, it would \emph{not} be robust in the worst case, where most of the batches could be corrupted. To ensure an accurate estimate, the batch size $K$ need to be sufficiently large, which makes the empirical mean of each batch vulnerable to corruptions.
A hand-wavy but intuitive argument goes as follows: since each $\matdotprod{v_j}{O_i}{v_j}$ is between $-1$ and $1$, and the $(\dims+1)$ factor in~\eqref{equ:classical-shadow-mean} amplifies the fluctuation, the empirical mean of each batch could shift by roughly $\gamma\dims$ if $\gamma$-fraction of the samples in the batch are corrupted. 
If at least $3/4$ of the batches are affected, and the fluctuations are in the same direction, the final median output would also shift by $\gamma\dims$.
We prove that this can indeed happen for a simple yet practical problem of fidelity estimation, an important application of shadow tomography. 
\begin{theorem}
\label{thm:lb-classical-shadow}
    Consider the fidelity estimation problem where we only have one observable $O=\qproj{\psi}$ for some pure state $\qbit{\psi}$, and the unknown state $\rho=\eye_\dims/\dims$ is the maximally mixed state. Let $K$ be the batch size for the MoM estimator. For $\ns$ sufficiently large, $\gamma\ge 1/K$, and $\dims\ge 1000$, there exists an adversary that corrupts $O(\gamma)$ fraction of the outcome vectors $\qbit{v_1}, \ldots, \qbit{v_\ns}$ and causes an error of at least $\Omega(\gamma\dims)$ for the MoM estimator~\cite{huang2020predicting}.
\end{theorem}
\begin{proof}
    When $\rho=\eye_\dims/\dims$, $\qbit{v}$ is exactly drawn from the Haar measure, and thus $\expect{\matdotprod{v}{O}{v}}=\expect{|\qdotprod{v}{\psi}|^2}=1/\dims$. By Levy's lemma,  for $\dims\ge 1000$,
    \[
    \probaOf{|\qdotprod{v}{\psi}|^2\le 1/3}\ge 0.99.
    \]
    Using Chernoff bound, for $\ns$ sufficiently large, $|\qdotprod{v_i}{\psi}|^2\le 1/3$ for at least 0.8 fraction of the $\qbit{v_i}$'s  with probability at least 0.99. In each batch, as long as $\gamma\ge 1/K$ where $K$ is the batch size, we can change $\gamma$-fraction of the vectors $\qbit{v_i}$ with smallest $|\qdotprod{v_i}{\psi}|^2$ to $\qbit{\psi}$, and thus the new value of $|\qdotprod{v_i}{\psi}|^2=1$. Therefore, the empirical mean of at least $0.8$ fraction of the batches increases by at least $\gamma\dims/3$, meaning that the median of the means $E'$ would also increase by at least the same value compared to the original MoM estimate $E$. Since error for the original estimate satisfies $|E-\Tr[O\rho]|\le \gamma\dims/12 $ when $\ns$ is sufficiently large, by triangle inequality the new estimate satisfies $E'\ge \Tr[O\rho]+\gamma\dims/4$, completing the proof.

    We note that the requirement on $\gamma\ge 1/K$ is very mild: for fidelity estimation, the batch size $K=\ns/\log(1/\delta)$ as suggested by \cite{huang2020predicting} where $\delta$ is the desired success probability to achieve a near-optimal error of $\sqrt{\log(1/\delta)/\ns}$.
\end{proof}

\begin{remark}
    One may argue that the real implementation of uniform POVM involves first sampling a unitary $U_i$ and then applying the corresponding basis measurement, thus the vector $\qbit{v_i}$ can only be chosen from the basis of $U_i$. However, in the adversarial setting, we argue that the random seed used to sample $U_i$ could be altered by the adversary, which would change $\qbit{v_i}$ arbitrarily. This captures potential errors in the sampling process, which is also extremely relevant in practice.
\end{remark}

\subsection{Our algorithm}
To obtain an adversarially robust algorithm, we make a very simple yet effective change to the classical shadows algorithm: instead of median-of-means, we use the truncated mean as the output. The error guarantee of the algorithm is well-known in robust statistics,
\begin{algorithm}
\caption{\texttt{TruncatedMean}($\bx, \gamma$)}
\label{alg:truncated-mean}
    \begin{algorithmic}
        \State \textbf{Input:} potentially corrupted samples $\bx =(x_1, \ldots, x_\ns)\in\R$, parameter $\gamma$.
        \State Ignore the largest and smallest $\gamma$-fraction of the samples. 
        \State\Return $\hat\mu=$the empirical average of the remaining samples.
    \end{algorithmic}
\end{algorithm}

\begin{theorem}[Section 1.4.2, \cite{diakonikolas2023algorithmic}]
\label{thm:truncated-mean-error}
    Let $X\in\R$ be a real-valued random variable with mean $\mu$ with $\expect{|X-\mu|^h}\le\sigma^h$. Given $\gamma$-corrupted samples $\bx = (x_1, \ldots, x_\ns)$ from $X$, \texttt{TruncatedMean}($\bx, 2\gamma$) (\cref{alg:truncated-mean}) outputs $\hat{\mu}$ such that $|\hat{\mu}-\mu|=O(\sigma\gamma^{1-1/h})$ with probability $1-\delta$ for some $\ns=O(\log(1/\delta)/\gamma^2)$.
\end{theorem}

For each observable $\Obs_i$, we simply apply the truncated mean estimator to estimate $\expectDistrOf{v\sim\cD(\rho)}{\matdotprod{v}{\Obs_i}{v}}$. Because the Haar measure is well concentrated, similar to hyper-contractivity proved in~\cite{AliakbarpourBCL2025robustquantum}, we can bound the higher order central moments of $\matdotprod{v}{O}{v}$ when $\qbit{v}\sim \mathcal{D}_{\rho}$ for all Hermitian matrices $O$,
\begin{theorem}
\label{thm:uniform-povm-moment-bound}
    Let $\qbit{v}\sim \mathcal{D}_{\rho}$ with probability density defined in~\eqref{equ:uniform-povm-density}. For all $h\ge 2$ and all Hermitian matrix $O\in \C^{\dims\times\dims}$, there exists a universal constant $C$ such that
    \[
    \expectDistrOf{v\sim\cD(\rho)}{(\bra{v}O\ket{v}-\expectDistrOf{v\sim\cD(\rho)}{\bra{v}O\ket{v}})^h}\leq \Paren{\frac{Ch\hsnorm{O}}{d+1}}^h.
    \]
\end{theorem}

This theorem essentially says that $\matdotprod{v}{O}{v}$ is sub-exponential with paramenter $C\hsnorm{O}/(\dims+1)$. However, the uniform POVM can only be approximated up to finite moments using unitary designs. Therefore, the finite moment bound is extremely useful for proving guarantees for algorithms using $t$-designs, where the $t$-th order moment matches exactly with the Haar measure.
\begin{theorem}
\label{thm:alg-t-design}
    Suppose $t$-design is used in the classical shadows algorithm and we apply $2\gamma$-truncated mean estimator for each observable. Under $\gamma$-corruption, the algorithm achieves an estimation error of $\bigO{t\gamma^{1-1/t}\max_{i\in[\nobs]}\hsnorm{O_i}}$ with probability at least $1-\delta$ using $\ns = \bigO{\log(\nobs/\delta)/\gamma^2}$ copies.
\end{theorem}
\begin{proof}
    Since $t$-design matches the moments of the Haar measure up to order $t$, by~\cref{thm:uniform-povm-moment-bound}, for each observable $\Obs_i$, the random variable $(\dims+1)\matdotprod{v}{O_i}{v}-\Tr[O_i]$ has a central $t$-th moment bounded by $\sigma^t$ with $\sigma = Ct\hsnorm{\Obs_i}$. By \cref{thm:truncated-mean-error}, the estimate for $\Obs_i$ has error at most $\eps = O(t\hsnorm{\Obs_i}\gamma^{1-1/t})$, and achieving a success probability of $1-\delta'$ requires $\ns =O( \log(1/\delta')/\gamma^2)$ samples. Choosing $\delta'=\delta/\nobs$ and applying union bound completes the proof.
\end{proof}
Setting $t=\log(1/\gamma)$, the error would be $\eps = \bigO{\gamma\log\frac{1}{\gamma}\max_{i\in[\nobs]}\{\hsnorm{\Obs_i}\}}$ and the sample complexity remains unchanged in terms of $\gamma,\delta$, and $\nobs$, which immediately proves \cref{thm:result-algorithm} in the results section.

For completeness, we provide a comprehensive proof of \cref{thm:truncated-mean-error,thm:uniform-povm-moment-bound,thm:alg-t-design} in~\cref{sec:shadow-tomograghy}. 

\paragraph{Implication for robust full state tomography} Our shadow tomography algorithm  immediately implies a sample efficient algorithm for full state tomography of rank-$\rk$ states using a reduction through a covering argument. 
\begin{lemma}[Full-state tomography reduction]
    Given a shadow tomography algorithm that achieves an accuracy of $\eps$ with probability $1-\delta$ using $\ns=\ns(\eps,\delta, \nobs)$ copies for $\nobs$ observables, we can output an estimate of the unknown state with error $O(\eps)$ in trace distance using $\ns(\eps,\delta, \nobs=\exp(\bigO{dr\log(r/\eps)}))$ copies.
\end{lemma}
\begin{proof}[Proof sketch]
    We sketch the idea of the reduction. First, we argue that there exists an $\eps$-cover $\mathcal{N}$ of rank-$\rk$ states of size $|\mathcal{N}|=\exp(\bigO{\dims\rk\log(r/\eps)})$. For every pair of states $\sigma_i, \sigma_j\in \mathcal{N}$, we can define the Holevo-Holstrom observable $0\preceq\Obs_{ij}\preceq \eye_\dims$ with rank at most $2\rk$ where difference in the expectation values of the observable corresponds to the trace distance,
\[
|\Tr[\sigma_{i}\Obs_{ij}]-\Tr[\sigma_{j}\Obs_{ij}]|=\tracenorm{\sigma_i-\sigma_j}/2.
\]
The number of observables is $\nobs \le |\mathcal{N}|^2=\exp(O(\dims r\log(r/\eps)))$.

Given an unknown rank-$\rk$ state $\rho$, by estimating all $\Tr[\rho O_{ij}]$ accurately, the Holevo-Holstrom observables enable us to compute some kind of ``surrogate trace distance'' with respect to all states in the cover $\mathcal{N}$. In particular, let $E_{ij}$ be the estimated expectation value of each observable $\Obs_{ij}$, the surrogate distance to some $\sigma_i\in \mathcal{N}$ is defined as
\[
d_i \eqdef\max_{j\in[|\mathcal{N}|]}|E_{ij}-\Tr[\sigma_jO_{ij}]|.
\]
Our estimate of $\rho$ is simply the state in the cover $\mathcal{N}$ that minimizes the surrogate distance. The $O(\eps)$ error guarantee can be obtained by repeated use of triangle inequality and that $O_{ij}$ maximizes $|\Tr[(\sigma_i-\sigma_{j})O]|$ for all $0\preceq O\preceq \eye_\dims$.
\end{proof}
We can apply our algorithm in~\cref{thm:result-algorithm}. Since $\hsnorm{O_{ij}}\le \sqrt{2r}$, we set $\eps = \tildeO{\gamma\sqrt{r}}$, which is the error guarantee of our algorithm.  Note that  to achieve a success probability of $1-\delta$, the number of copies used by our robust shadow tomography algorithm is
\[
\ns = \bigO{\frac{1}{\gamma^2}\log\frac{M}{\delta}}=\tildeO{\frac{\dims r}{\gamma^2}}.
\]
Up to log factors, this matches the lower bound for single copy tomography that achieves the same error of $\eps =\gamma\sqrt{r}$ without corruption.  
Despite using $\exp(\dims \rk)$ post-processing time, we nearly resolve the sample complexity of robust state tomography. We leave the problem of further improving computational efficiency as future work.

\subsection{Lower bound}
\label{sec:overview-lower}
To prove the error lower bound, we construct a one-versus-many testing problem which an accurate shadow tomography algorithm must be able to distinguish.
Specifically, we construct a set of traceless observables $\Obs_1, \ldots, \Obs_{\nobs}$ with eigenvalues either $+1$ or $-1$. If an algorithm can learn all expectation values with accuracy $\eps$, it should distinguish between the following two cases,
\[
H_0: \rho = \qmm\eqdef\frac{\eye_\dims}{\dims},\quad H_1:\sigma_i\eqdef \frac{\eye_\dims+3\eps\Obs_i}{\dims},i\sim[\nobs].
\]

We then construct an adversary that, in expectation, changes at most $\gamma$-fraction of the outcomes to make outcomes from $H_1$ look exactly the same as the second case $H_2$ where the unknown state is uniformly sampled from one of $\frac{\eye_\dims+3\eps\Obs_i}{\dims}$.
Let $\p_{\rho,\POVM}$ be the outcome distribution when applying a measurement $\POVM$ to $\rho$.
Then the adversary could use a coupling strategy to corrupt the outcomes. A coupling between two distributions $\p_1, \p_2$ can be viewed as a randomized function $F$ such that if $X\sim \p_1$, $F(X)$ is exactly distributed as $\p_2$. 
Standard results in statistics~\cite[Theorem 2.12]{Hollander12} say that there exists a coupling for all pairs of distributions such that $\probaOf{F(X)\ne X}=\totalvardist{\p_1}{\p_2}$. 

When receiving outcomes $x_1, \ldots, x_\ns$, which are obtained by measuring $\qmm$ with $\POVM_{1}, \ldots, \POVM_{\ns}$, the adversary does the following,
\begin{enumerate}
    \item Sample $i\in[\nobs]$ uniformly.
    \item For each outcome $x_j$, apply a coupling between $\p_{\qmm, \POVM_j}$ and $\p_{\sigma_i, \POVM_j}$.
\end{enumerate}
By the definition of coupling, \emph{in expectation}, the fraction of outcomes changed by the adversary is at most
\begin{equation}
\label{equ:expect-total-var-dist}
    \gamma = \frac{1}{\ns\nobs}\sum_{j=1}^{\ns}\sum_{i=1}^{\nobs}\totalvardist{\p_{\qmm,\POVM}}{\p_{\sigma_i,\POVM}}.
\end{equation}

Moreover, the outcomes would be distributed exactly as if $\rho$ is randomly picked from $\sigma_i$'s, so the algorithm cannot distinghuish the two cases, and thus must incur an error of at least $\eps$. 
Therefore it suffices to upper bound the expected total variation distance~\eqref{equ:expect-total-var-dist} in terms of $\eps, \dims$ and $\nobs$. 

We show that there exists a set of $\nobs$ observables such that $\gamma=\tildeO{\eps/\sqrt{\min\{\nobs, \dims\}}}$ for all measurements. Rearranging the terms gives a lower bound for $\eps$ that matches our lower bound in~\cref{thm:result-lower}.

For direct measurements on the observables, (which corresponds to $\{(\eye_\dims+\Obs_i)/2, (\eye_\dims-\Obs_i)/2\}$ for our choice of $O_i$'s), we can find a different set of observables and \eqref{equ:expect-total-var-dist} is bounded as $\gamma=\bigO{\eps/\nobs}$, which proves the $\eps=\Omega(\gamma\nobs)$ lower bound in~\cref{thm:naive}.

A small technical issue is that the adversary only satisfies the $\gamma$-fraction constraint in expectation. 
However, by Markov's inequality, with probability at least $1/2$, the adversary should not change more than $2\gamma\ns$ outcomes. 
If we force the adversary to stop when $2\gamma\ns$ outcomes have been changed, the resulting outcome distribution would still be close to that in the case of $H_1$ in total variation distance, which is still hard to distinguish for all algorithms. 
We provide the full proof in~\cref{sec:lower} and also refer the readers to \cite[Section 4.1]{AliakbarpourBCL2025robustquantum}.
\label{sec:overview-lb}

\section{Experiments and practical implications}
We conduct a simulation experiment on classical computers and compare our approach and that of \cite{huang2020predicting}. Using the original uniform POVM, we find that a truncation-based approach is indeed more robust to adversarial noise than the standard classical shadows using MoM. 
This provides strong support for our theoretical results and demonstrates that our proposed approach could have a significant practical impact.

\paragraph{Experiment setup}
We test our algorithm for fidelity estimation. We have an unknown state $\nqubits=5$ qubit pure state $\qbit{\phi}$, and we design $\nobs$ random pure states $\mathcal{S}=\{\qbit{\psi_i}\}_{i\in [\nobs]}$ where the ground truth fidelity  $|\qdotprod{\psi_i}{\phi}|^2=0.9$. 
The goal is to estimate the fidelity of $\qbit{\phi}$ with respect to all the $\qbit{\psi_i}$'s. The number of copies is $\ns=10^{4}$, which is chosen to target an accuracy level of the order $0.01$. 

The noise simulates the worst-case adversary for the classical shadow: we first pick a state from the set $\mathcal{S}$, say $\qbit{\psi_0}$, and then independently with probability $\gamma\in[0, 0.02]$ changes the outcome vector $\qbit{v_i}$ to $\qbit{\psi_0}$. For the state $\qbit{\psi_0}$, the MoM estimator would have an error at least $\Omega(\gamma\dims)$ using the same argument as~\cref{thm:lb-classical-shadow}. We also argue that it is practically reasonable to only experiment with $\gamma$ with roughly the same order as the targeted accuracy. Indeed, if an experimentalist believes that the probability of corruption is much larger than the desired accuracy, then it makes no sense to continue the experiment and it would be wise to try to reduce the level of noise in the system. 

\paragraph{Result} 
We present the result for $\nobs=1$ and $\nobs=62$ in~\cref{fig:experiment-main}. The latter value is chosen so that $\nobs\ge \dims$ which falls into the regime where classical shadows beats direct measurement approach. Experiments are repeated 5 times with standard deviation shown in the figure as error bars.
\begin{figure}
    \centering
    \includegraphics[width=0.46\linewidth]{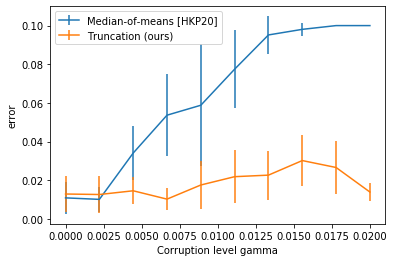}
    \includegraphics[width=0.46\linewidth]{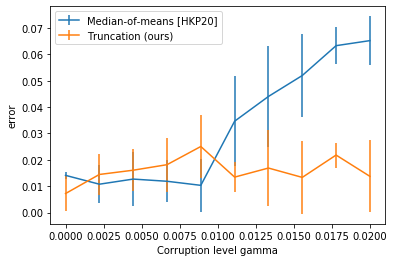}
    \caption{Comparison of MoM and truncation-based classical shadows for different number of observables. \textbf{Left:} $\nobs=1$. \textbf{Right:} $\nobs=62$.}
    \label{fig:experiment-main}
\end{figure}
We can observe a general trend that for small $\gamma$ our algorithm performs on-par with \cite{huang2020predicting} which we showed does not have worst-case robustness guarantees. 
Our truncation algorithm exhibits much stronger robustness when $\gamma$ is large, where we significantly outperform the MoM estimator. 
Our advantage is more evident for $\nobs=1$, which might be more relevant for fidelity estimation as we often create a quantum system to generate a specific state of interest. These results provide a strong support for our theoretical claims.


\paragraph{Practical implications} Our simulation experiment together with our theoretical bounds convey a clear and powerful message: a simple truncation step protects the widely used classical shadows algorithm from small-probability worst-case noise, with provable theoretical guarantees, and without using any prior knowledge about the noise structure. 
Practitioners may already be adopting this approach when processing their data, but now they should have more peace of mind given our theoretical guarantees and simulation results.
Although only tested on classical computers using the ideal uniform POVM, we have confidence that our approach is also effective for practical implementations of the uniform POVM on real quantum computers.

We also emphasize that our algorithm is not a replacement, but rather an addition to existing approaches, as its simplicity could be combined with many types of algorithms, including the MoM used by~\cite{huang2020predicting}. 
Furthermore, our idea could be combined with existing model-specific robust shadow tomography algorithms~\cite{ChenYuZengFlammia2020robustshadow,hu2025demonstration,Brieger2025stability}, which could simultaneously provide better guarantees for specific noise structures and be robust if the assumptions are not satisfied for some small probability.

\section{Preliminaries}

\subsection{Quantum state and POVM}
\paragraph{Complex matrices}
Let $A,B\in\C^{\dims\times\dims}$. Inner product between $A$ and $B$ is defined as $\hdotprod{A}{B}\eqdef\Tr[A^\dagger B]$. 
We denote the set of all $\dims\times\dims$ Hermitian matrices by $\Herm{\dims}$.
For Hermitian matrices $A,B$, $\hdotprod{A}{B}=\hdotprod{B}{A}\in \R$. Thus the subspace of Hermitian matrices $\Herm{\dims}$ is a \textit{real} Hilbert space (i.e. the associated field is $\R$) with the same matrix inner product.

\paragraph{Schatten norms} Let $\Lambda=(\lambda_1, \ldots, \lambda_\dims)\ge 0$ be the \emph{singular values} of a matrix $A$. {For Hermitian matrices, the singular values are the absolute values of the eigenvalues.} Then for $p\ge 1$, the \emph{Schatten $p$-norm} is defined as 
$
\|A\|_{S_p}\eqdef \|\Lambda\|_p
$. Important examples include trace norm $\tracenorm{A}\eqdef\|A\|_{S_1}$, Hilbert-Schmidt norm $\hsnorm{A}\eqdef\|A\|_{S_2}=\sqrt{\hdotprod{A}{A}}$, and operator norm $\opnorm{A}\eqdef\|A\|_{S_\infty}=\max_{i=1}^\dims\lambda_i$. 

\paragraph{Quantum states} We use $\qbit{\psi}$ to represent a complex vector in $\C^{\dims}$. For $j\in\{0, \ldots, \dims-1\}$, $\qbit{j}$ denotes the vector with 1 at the $j$th coordinate and 0 everywhere else (assuming the coordinates are numbered starting from 0). $\qadjoint{\psi}\eqdef(\qbit{\psi})^\dagger$ is a row vector. $\qdotprod{\psi}{\phi}$ is the inner product of $\qbit{\psi}$ and $\qbit{\phi}$. 

A quantum system with $\nqubits$ qubits has dimensionality $\dims=2^{\nqubits}$. It can be described by \emph{density matrix} $\rho\in\Herm{\dims}$, which is a positive-semidefinite Hermitian matrix with $\Tr[\rho]=1$. A special case is the maximally mixed state $\qmm\eqdef \eye_\dims /\dims$. 

The set of $d$-dimensional quantum states of rank $r$ is denoted by $\mathcal{S}_{d,r}$. In particular, states with $r=1$ is called pure states, and we often use the vector form , e.g. $\qbit{\psi}\in\C^{\dims}$ to denote pure states. States with $r>1$ are called mixed states.

The trace distance between quantum states is $\norm{\cdot}_{\Tr}\eqdef\frac{1}{2}\tracenorm{\cdot}$. As they only differ by a constant, we sometimes use them interchangeably. Another commonly used matrix is called \emph{fidelity}. For a mixed state $\rho$ and pure state $\qbit{\psi}$, the fidelity is defined as $F(\rho,\psi)=\matdotprod{\psi}{\rho}{\psi}$.\footnote{Some literature define it as the square root of what we have.}
There is a more general definition for fidelity between mixed states, but it will not be used in our work.

\paragraph{Measurements} The most general formulation of measurements is called the \emph{positive operator-valued measure} (POVM). Let $\outset$ be an outcome set. Then a POVM $\POVM=\{M_x\}_{x\in \outset}$ is a set of matrices indexed by the outcomes where $M_x$ is p.s.d. and $\sum_{x\in \outset}M_x=\eye_\dims$. Let $X$ be the outcome of measuring $\rho$ with $\POVM$, then the probability observing $x\in\outset$ is given by the \emph{Born's rule},
\[
\probaOf{X=x}=\Tr[\rho M_x].
\]
This definition can be extended to an infinite outcome set.

\subsection{Probability distances}
Let $\p$ and $\q$ be distributions over a finite domain $\mathcal{X}$. We define $\ell_p$ distance as 
\[\norm{\p-\q}_p\eqdef\Paren{\sum_{x\in\mathcal{X}}{|\p(x)-\q(x)|^p}}^{1/p}.\]
We define the \emph{total variation distance},
\[
\totalvardist{\p}{\q}\eqdef\sup_{S\subseteq\mathcal{X}}(\p(S)-\q(S))=\frac{1}{2}\sum_{x\in\mathcal{X}}|\p(x)-\q(x)|.
\]
Note that $\totalvardist{\p}{\q}=\frac{1}{2}\norm{\p-\q}_1$.
The KL-divergence  is defined as
\[
\kldiv{\p}{\q}\eqdef\sum_{x\in\mathcal{X}}\p(x)\log\frac{\p(x)}{\q(x)}.
\]
The chi-square divergence is defined as follows,
\[
\chisquare{\p}{\q}\eqdef \sum_{x\in\mathcal{X}}\frac{(\p(x)-\q(x))^2}{\q(x)}.
\]
By Pinsker's inequality and concavity of the logarithm function,
\[
2\totalvardist{\p}{\q}^2\le \kldiv{\p}{\q}\le \chisquare{\p}{\q}.
\]

All the definitions can be extended to general probability measures.

\subsection{Haar measure}
We use $\Haar{d}$ to denote the Haar measure over $\dims\times\dims$ unitary matrices, which can be viewed as a uniform distribution over unitary matrices.
Let $\Sp^{\dims-1}$ denote the complex unit sphere in $\C^{\dims}$. 
The Haar measure induces a unique unitarily invariant measure on $\Sp^{\dims-1}$, which we denote as $\dm v$. It is essentially the uniform distribution on the complex unit sphere.
Understanding the Haar measure requires sophisticated representation theoretic techniques, which is not within the scope of this paper. We only display the necessary definitions and results.

\paragraph{Permutation and cycles}
A permutation $\pi:[n]\mapsto[n]$ is a bijection over $[n]$. Let $\Sim_n$ be the set of all permutations over $[n]$. Every permutation can be decomposed into cycles. For example, the permutation $\pi=(2, 3, 1, 5, 4)$ has just one cycle
\[
(4,5), (1,2,3).
\]
We use $\cycle(\pi)$ to denote the set of cycles in $\pi$. For a cycle $c\in \cycle(\pi)$, $|c|$  denotes the length of the cycle. In the example above $|(1,2,3)|=3$ and $|(4,5)|$=2.

\paragraph{Haar measure moments}
Let $\qbit{u}\in \C^{\dims}$ be a unit vector drawn uniformly from the unit sphere. The $\ab$-th order moment $\qproj{u}^{\otimes\ab}$ can be computed using Schur-Weyl duality, \cite[Eq.(14)]{guctua2020fast}.
\begin{equation}
    \expectDistrOf{\qbit{u}\sim\Haar{\dims}}{\qproj{u}^{\otimes\ab}}=\binom{\dims+k-1}{k}^{-1}P_{\text{Sym}^{(k)}}.
    \label{equ:haar-k-moment}
\end{equation}
Here $P_{\text{Sym}^{(k)}}=\frac{1}{k!}\sum_{\pi\in \Sim_{k}}P_{\pi}$ is the projection matrix onto the symmetric subspace of $(\C^{\dims})^{\otimes k}$, where $P_\pi$ is a permutation operator defined as
\[
P_\pi\qbit{\psi_1}\otimes\qbit{\psi_{\ab}}=\qbit{\psi_{\pi^{-1}(1)}}\otimes\cdots\otimes\qbit{\psi_{\pi^{-1}(\ab)}}.
\]
Permutation operators are unitary with $P_{\pi}^\dagger = P_{\pi}^{-1}=P_{\pi^{-1}}$.
The following equation is handy when analyzing our algorithms.
\begin{equation}
    \expectDistrOf{\qbit{u}\sim\Haar{\dims}}{\matdotprod{u}{M}{u}^k}=\binom{\dims+k-1}{k}^{-1}\frac{1}{k!}\sum_{\pi\in\Sim_k}\prod_{c\in \cycle(\pi)}\Tr[M^{|c|}].
    \label{equ:haar-trace-moment}
\end{equation}

\paragraph{Unitary/spherical designs} It is physically impossible to uniformly sample from the Haar measure on quantum systems. Thus, people often use $t$-designs as a surrogate, which is a finite set of unitaries/unit vectors such that uniformly sampling from the set has matching $t$-th order moment as the Haar measure.

\subsection{Covering}

An $\eps$-covering of a set $K$ for the norm $\|\cdot \|$ is a subset $\mathcal{N} \subset K$ such that,
\begin{itemize}
    \item $\mathcal{N}$ is finite, 
    \item for every $x\in K$, there exists $y\in \mathcal{N}$ such that $\|x - y\| \le \eps.$
\end{itemize}
 We denote by $N(K, \|\cdot \|, \eps)$ the minimal cardinality of an $\eps$-covering of $K$ for the norm $\|\cdot \|$ and the closeness parameter $\eps$. 

\section{Lower bound}
 \label{sec:lower}
\subsection{Lower bound for general measurements and observables}
This section proves our main lower bound result~\cref{thm:result-lower}. \\
Let $\dims$ be a power of 2. We construct a set of observables $\Obs_1, \ldots, \Obs_{\nobs}$ with $\Tr[\Obs_i]=0$ and eigenvalues of $\pm1$.
Let $\eps\le 1/3$. A shadow tomography algorithm that achieves an error of $\eps$ must be able to solve the following one-versus-many hypothesis testing problem,
\begin{align}
\label{equ:hard-case-general}
H_0: \qmm\eqdef\frac{\eye_\dims}{\dims}, \quad H_1: \sigma_i =\frac{\eye_\dims+3\eps\Obs_i}{\dims}, i\sim [\nobs].
\end{align}
This is because $\Tr[\Obs_i\qmm]=0$ for all $\Obs_i$, while $\Tr[\Obs_i \sigma_i] = \Tr[\Obs_i(\eye_\dims+3\eps\Obs_i)/\dims]=3\eps$. 
Thus, if the unknown state $\rho=\qmm$, all estimates $\{E_i\}_{i\in[\nobs]}$ of an $\eps$-accurate shadow tomography algorithm should be within $[-\eps,\eps]$ with probability at least $1-\delta$. 
On the other hand, if $\rho=\sigma_i$ for some $i\in [\nobs]$, at least one of the estimates $E_j$ should be at least $2\eps$: indeed, the estimate for $O_i$ should satisfy $|\Tr[\sigma_i O_i]-E_{i}|\le \eps$, which implies $E_i\ge 2\eps$.

When $\gamma$-fraction of the samples are corrupted, outcomes that come from $H_0$ could look much closer to $H_1$ than without corruption. 
Formally, \cite{AliakbarpourBCL2025robustquantum} characterizes the hardness of testing under adversarial corruption using the earth-mover distance.
\begin{definition}
\label{def:emd}
    Let $\bfP_1,\bfP_2$ be two distributions over $\cX^{\ns}$ (note that they may not necessarily be product distributions). Let $\pi(\bfP_1, \bfP_2)$ be the set of couplings between the two distributions. The earth-mover distance is
    \[
    \EMD{\bfP_1}{\bfP_2}\eqdef \min_{\Pi\in \pi(\bfP_1,\bfP_2)}\expectDistrOf{(X^\ns, Y^{\ns})\sim \Pi}{\ham{X^{\ns}}{Y^{\ns}}}.
    \]
\end{definition}

\begin{lemma}[From Lemma 4.6, \cite{AliakbarpourBCL2025robustquantum}]
\label{lem:emd}
    Let $\rho\in \C^{\dims\times\dims}$ be a fixed state and $\mathcal{D}$ be a distribution of states. If for all measurements $\POVM^{\ns}=(\POVM_1, \POVM_2, \ldots, \POVM_\ns)$, 
    \[
    \EMD{\expectDistrOf{\sigma\sim\mathcal{D}}{\p_{\sigma}^{x^n}}}{\p_{\rho}^{x^n}}\le \frac{\gamma \ns}{2},
    \]
    then there does not exist an algorithm that can distinguish between $H_0:\sigma=\rho$ and $H_1:\sigma\sim \mathcal{D}$ with probability at least 0.8.
\end{lemma} 

The following helper lemma is implicit in the proof of \cite[Theorem 4.9]{AliakbarpourBCL2025robustquantum}.
\begin{lemma}
\label{lem:emd-tv-distance}
Let $\rho\in \C^{\dims\times\dims}$ be a fixed state and $\mathcal{D}$ be a distribution of states.
    \[
    \EMD{\expectDistrOf{\sigma\sim\mathcal{D}}{\p_{\sigma}^{x^n}}}{\p_{\rho}^{x^n}}\le \sum_{i=1}^\ns\expectDistrOf{\sigma\sim\cD}{\totalvardist{\p_{\sigma}^i}{\p_{\rho}^i}}.
    \]
\end{lemma}
Therefore, it suffices to upper bound the expected total variation distance for each copy. 
Without loss of generality, assume that the measurement applied to a particular copy has the form of $\POVM=\{\dims w_x\qproj{\psi_x}\}_{x\in\cX}$, and let $\p_\rho$ be the outcome distribution when $\POVM$ is applied to $\rho$.
We use the shorthand $\p_i$ to denote $\p_{\sigma_i}$.
We can thus bound the expected total variation distance as follows,
\begin{align}
    \expectDistrOf{i\sim [\nobs]}{\totalvardist{\p_{i}}{\p_{\qmm}}} &\le \expectDistrOf{i\sim [\nobs]}{\sqrt{\frac{1}{2}\chisquare{\p_{i}}{\p_{\qmm} }}}\label{equ:lb-general-pinsker}\\
    &\le\sqrt{\frac{1}{2}\expectDistrOf{i\sim [\nobs]}{\chisquare{\p_{i}}{\p_{\qmm}} }}\label{equ:lb-general-concave}.
\end{align}
\eqref{equ:lb-general-pinsker} is by Pinsker's inequality and the relation between KL and $\chi^2$ divergence. \eqref{equ:lb-general-concave} is due to concavity of square root.

By Born's rule, $\p_{\qmm}(x)=w_x$ and $\p_i(x)=w_x(1+3\eps\matdotprod{\psi_x}{O_i}{\psi_x})$. Therefore,
\[
\chisquare{\p_{i}}{\p_{\qmm}}=\sum_{x}\frac{(\p_i(x)-\p_{\qmm}(x))^2}{\p_{\qmm}(x)}=9\eps^2\sum_{x}w_x \matdotprod{\psi_x}{O_i}{\psi_x}^2.
\]
Proceeding from \eqref{equ:lb-general-concave},
\begin{align}
\label{equ:lb-almost-final}
    \expectDistrOf{i\sim [\nobs]}{\totalvardist{\p_{i}}{\p_{\qmm}}}\le \frac{3\eps}{\sqrt{2}}\sqrt{\expectDistrOf{i\sim[\nobs]}{\sum_{x}w_x \matdotprod{\psi_x}{O_i}{\psi_x}^2}}=\frac{3\eps}{\sqrt{2}}\sqrt{\sum_{x}w_x\expectDistrOf{i\sim[\nobs]}{ \matdotprod{\psi_x}{O_i}{\psi_x}^2}}.
\end{align}
Our next step is to show that there exists a set of observables such that the above quantity is small for all measurements. 
We can directly apply a result on randomly chosen observables from \cite{ChenCH021},
\begin{lemma}[Lemma 5.6, \cite{ChenCH021}]
\label{lem:max-obs-square}
    There exists a set of observables $O_1, \ldots, O_{\nobs}$ such that for all unit vector $\qbit{\psi}\in \C^\dims$,
    \[
    \frac{1}{\nobs}\sum_{i=1}^{\nobs}\matdotprod{\psi}{O_i}{\psi}^2\le \frac{2304\pi^3\log(1+2\min\{\nobs, \dims\})}{\min\{\nobs, \dims\}}+\frac{4}{\min\{\nobs, \dims\}}+\frac{1}{\dims+1} =\bigO{\frac{\log\nobs}{\nobs}+\frac{\log\dims}{\dims}}.
    \]
\end{lemma}
\begin{remark}
    The original statement has $\nobs$ in place of all the $\min\{\nobs, \dims\}$. We believe that the original proof has a minor error which does not extend to the case when $M\ge \dims$ (the authors of \cite{ChenCH021} only claim that when $M\ge \dims$ a lower bound of $M/\log(M)$ would not hold, but did not seem to prove that a lower bound of $\dims$ would hold in this regime). We propose a simple fix that leads to a slightly different bound.
As a result, the sample complexity lower bound for single-copy shadow tomography would be $\Omega(\min\{\nobs/\log\nobs, \dims/\log\dims\}/\eps^2)$ instead of $\Omega(\min\{\nobs/\log\nobs, \dims\}/\eps^2)$ as claimed by \cite{ChenCH021}.
\end{remark}
Since $i$ is chosen uniformly from $[\nobs]$, combining with \eqref{equ:lb-almost-final} and \cref{lem:emd-tv-distance}, we have
\[
\EMD{\expectDistrOf{\sigma\sim\mathcal{D}}{\p_{\sigma}^{x^n}}}{\p_{\rho}^{x^n}}\le \sum_{i=1}^\ns\expectDistrOf{\sigma\sim\cD}{\totalvardist{\p_{\sigma}^i}{\p_{\rho}^i}}\le \bigO{\frac{\ns\eps}{\sqrt{\min\{\nobs/\log\nobs, \dims/\log\dims\}}}}.
\]
By \cref{lem:emd}, we must have $\gamma\le c\eps/\min\{\nobs/\log\nobs, \dims/\log\dims\}$ for some universal constant $c$, otherwise testing between the two hypotheses would not be possible. Rearranging the terms leads to the error lower bound of $\eps=\Omega(\gamma\sqrt{\min\{\nobs/\log\nobs, \dims/\log\dims\}})$.

\subsection{Lower bound for directly measuring the observables}
\label{sec:lb-direct}
\begin{theorem}
    Let $\Obs_1, \ldots, \Obs_{\nobs}$ a set of Pauli observables.   If a non-adaptive algorithm measures each copy with one of the measurements $\POVM_i\eqdef\{\frac{\eye_\dims - O_i}{2}, \frac{\eye_\dims +O_i}{2}\}$, it must incur an error of $\eps=\Omega(\gamma \nobs)$ under $\gamma$-adversarial corruption.
\end{theorem}
\begin{proof}
    The hard states are defined as~\eqref{equ:hard-case-general}, and recall that in this definition $\sigma_i=(\eye_\dims + 3\eps\Obs_i)/\dims$.
    Given $\POVM_j$, we have $\p_i=\bernoulli{\frac{1}{2}}$ for $i\ne j$ and $\bernoulli{\frac{1+3\eps}{2}}$ if $i=j$. Furthermore, $\p_{\qmm}=\bernoulli{1/2}$. Thus, 
    \[
    \totalvardist{\p_i}{\p_\qmm}= 3\eps\indic{i=j}.
    \]
    We can directly compute the expected total variation distance as required by~\cref{lem:emd-tv-distance} instead of going through Pinsker's inequality.
    \[
    \expectDistrOf{i\sim[\nobs]}{\totalvardist{\p_i}{\p_\qmm}}=\frac{3\eps}{\nobs}.
    \]
    Combining with~\cref{lem:emd}, we must have $\gamma\le 3\eps/\nobs$, which proves the theorem.
\end{proof}
\begin{remark}
    For $\nobs\le 3^\nqubits$, the theorem also applies to the case when the measurements are chosen as the basis of each observable, i.e., for $\Obs_i=\sum_{j=1}^\dims \lambda_j^{(i)}\qproj{\psi_j\supparen{i}}$, we use the projective measurement $\POVM_i=\{\qproj{\psi_j\supparen{i}}\}_{j\in[\dims]}$. We can choose observables $\Obs_1, \ldots, \Obs_{\nobs}$ as Pauli observables with no idendity components (e.g. $X^{\otimes\nqubits}$). Since Pauli observables are orthogonal, for each $\POVM_j$ we also have $\totalvardist{\p_i}{\p_\qmm}= 3\eps\indic{i=j}$ (here $\p_i$ and $\p_{\qmm}$ are discrete distributions over $\dims$ elements). The remaining steps follows similarly.
\end{remark}

\subsection{Proof of~\cref{lem:max-obs-square}}
We choose $O_i=U_i^\dagger Z U_i$ where $Z$ is a diagonal matrix with half $+1$ entries and half $-1$ entries, and $U_i$ independently sampled from the Haar measure.
Equation (106) of~\cite{ChenCH021} ensures that for all $0<\eta< 1$ and $t>0$, 
\[
\probaOf{\exists \qbit{\psi}, \left|\frac{1}{\nobs}\sum_{i=1}^{\nobs}\matdotprod{\psi}{U_i^\dagger Z U_i}{\psi}-\frac{1}{\dims+1}\right|\ge t+4\eta}\le 2(1+2/\eta)^{2\dims}\exp\Paren{-{\nobs}\min\left\{\frac{t^2}{\lambda^2},\frac{t}{\lambda}\right\}},
\]
where $\lambda = 288\pi^3/\dims$. The next step is to choose $\eta$ and $t$ appropriately to ensure the above probability is upper bounded by constant strictly less than 1. We replace Equation (107) of~\cite{ChenCH021} to the following,
\[
\eta = \frac{1}{\min\{\nobs, \dims\}}, \quad t=\lambda\cdot \frac{8\dims\log(1+2\min\{\nobs,\dims\})}{\min\{\nobs, \dims\}}.
\]
This choice ensures that $t/\lambda\ge 1$ and thus $\min\left\{\frac{t^2}{\lambda^2},\frac{t}{\lambda}\right\}=t/\lambda$ , which is crucial to canceling the $\nobs$ factor in the exponent. Substituting the new choice of parameters, we have
\begin{align*}
2(1+2/\eta)^{2\dims}\exp\Paren{-{\nobs}\min\left\{\frac{t^2}{\lambda^2},\frac{t}{\lambda}\right\}}&=2\exp\Paren{2\dims\log(1+2\min\{\nobs, \dims\})-\frac{M}{\min\{\nobs, \dims\}}8\dims(1+2\min\{\nobs, \dims\})}\\
&\le 2\exp \Paren{-6\dims\log(1+2\min\{\nobs, \dims\})}\\
&\le 2\exp(-2).
\end{align*}
By a probabilistic argument, there must exists a choice of $U_1, \ldots U_\nobs$ such that for all pure state $\qbit{\psi}$
\[
\frac{1}{\nobs}\sum_{i=1}^{\nobs}\matdotprod{\psi}{U_i^\dagger Z U_i}{\psi}\le t+4\eta +\frac{1}{\dims+1}=\bigO{\frac{\log\nobs}{\nobs}+\frac{\log\dims}{\dims}},
\]
completing the proof. 

\section{Shadow tomography algorithm}\label{sec:shadow-tomograghy}
\ynote{Put the proof for truncated mean in a separate subsection}
In this section we are going to show the robust shadow tomography algorithm.

\newtheorem*{restatemainthm}{\Cref{thm:result-algorithm} (Restate)}
\begin{restatemainthm}
    Under $\gamma$-adversarial corruption, there exists a shadow tomography algorithm such that for all observables $\Obs_1, \ldots, \Obs_{\nobs}$, achieves an error of at most $\eps=\bigO{\gamma\log(1/\gamma){\max_{j\in[\nobs]}\hsnorm{O_j}}}$ with probability at least $1-\delta$. The number of copies used by the algorithm is at most
    \[
    \ns = O\left({\frac{\max_{i\in[\nobs]}\hsnorm{O_i}^2}{\eps^2}\log^2(\frac{1}{\eps})\log\frac{\nobs}{\delta}}\right)
    =O\left(\frac{1}{\gamma^2}\log\frac{\nobs}{\delta}\right).
    \]
\end{restatemainthm}
\subsection{Bounded central moment}

\begin{lemma}
\label{lemma-h-central-moment}
For all integers $h\geq 2$, there exist a constant $C'>0$ for all Hermitian matrix $O \in \HH_d$ such that
$$\expectDistrOf{v\sim\cD(\rho)}{(\bra{v}O\ket{v}-\expectDistrOf{v\sim\cD(\rho)}{\bra{v}O\ket{v}})^h}\leq \frac{(C'h)^h}{(d+1)^h} \hsnorm{O}^h$$
\end{lemma}

\begin{proof}
Let $\overline{O}=O-\frac{\Tr[O]}{d}\eye_\dims$ be the traceless part of the Hermitian matrix $O$. We can show that
$$\bra{v}O\ket{v}-\expectDistrOf{v\sim\cD(\rho)}{\bra{v}O\ket{v}}=\bra{v}\overline{O}\ket{v}-\expectDistrOf{v\sim\cD(\rho)}{\bra{v}\overline{O}\ket{v}}.$$
Notice that the function $f(x)=x^h$ is convex on $x\in (-\infty,\infty)$ for all even integer $h\ge 2$. By Jensen’s inequality, $f(\frac{a+b}{2})\le \frac{f(a)+f(b)}{2}$, we have
$$\left[\bra{v}O\ket{v}+\expectDistrOf{v\sim\cD(\rho)}{\bra{v}O\ket{v}}\right]^h \le 2^{h-1}\left[\bra{v}O\ket{v}^h+(\expectDistrOf{v\sim\cD(\rho)}{\bra{v}O\ket{v}})^h\right]$$
Since $h$ is even, we have for all $a,b \in \mathbb{R}$, $(a-b)^h\leq (|a|+|b|)^h \leq 2^{h-1}(|a|^h+|b|^h)=2^{h-1}(a^h+b^h).$ By taking the expectation with respect to $v\sim\cD(\rho)$, we have for all even integer $h \ge 2$,
\begin{equation}
\label{eq:h-moment-convexity-1}
\expectDistrOf{v\sim\cD(\rho)}{
\left[\bra{v}O\ket{v}-\expectDistrOf{v\sim\cD(\rho)}{\bra{v}O\ket{v}}\right]^h} \le 2^{h-1}\,\expectDistrOf{v\sim\cD(\rho)}{\bra{v}O\ket{v}^h+(\expectDistrOf{v\sim\cD(\rho)}{\bra{v}O\ket{v}})^h}
\end{equation}
The first term on the RHS of (\ref{eq:h-moment-convexity-1}) can be bounded using the hypercontractivity of $\cD(\rho)$,
\begin{theorem}[Theorem 5.7~\cite{AliakbarpourBCL2025robustquantum}] 
   \label{thm:hypercontractive} 
For all integers $h\ge 2$, $\cD(\rho)$ satisfies $C$-hypercontractivity, i.e., there exists constant $C>0$ for all Hermitian matrix $O$, 
\[
\expectDistrOf{v\sim\cD(\rho)}{\matdotprod{v}{O}{v}^{h}}^2\le \frac{(Ch)^{2h}}{\dims^{2h}}(\Tr[O^2]+\Tr[O]^2)^h.
\]
\end{theorem}

Using this result, we have
\begin{align*}
2^{h-1}\,\expectDistrOf{v\sim\cD(\rho)}{\bra{v}O\ket{v}^h}&\le 2^{h-1}\frac{(Ch)^h}{d^h}(\Tr[\overline{{O}}^2]+\Tr[\overline{O}]^2)^{h/2}\\
& = 2^{h-1}\frac{(Ch)^h}{d^h}\left(\Tr[\overline{{O}}^2]^{1/2}\right)^{h}\\
& = 2^{h-1}\frac{(Ch)^h}{d^h} \hsnorm{\overline{O}}^h.
\end{align*}
The second term on the RHS of (\ref{eq:h-moment-convexity-1}) can be calculated:
\begin{align*}
2^{h-1}\,(\expectDistrOf{v\sim\cD(\rho)}{\bra{v}O\ket{v}})^h=2^{h-1}\, \left(\frac{\Tr[\overline{O}]+\Tr[{\overline{O}\rho}]}{(d+1)}\right)^h &=2^{h-1}\, \left(\frac{\Tr[{\overline{O}\rho}]}{(d+1)}\right)^h\\
& \le \frac{2^{h-1}}{(d+1)^h}\, \|\rho\|_1^h \cdot \|\overline{O}\|_\infty^h\\
& = \frac{2^{h-1}}{(d+1)^h}\, \|\overline{O}\|_\infty^h\\
& \le \frac{2^{h-1}}{(d+1)^h}\, \|\overline{O}\|_{\text{HS}}^h.
\end{align*}
Using the triangular inequality $\hsnorm{\overline{O}}\le \hsnorm{O}+\hsnorm{\frac{\Tr(O)}{d}\eye_\dims}$ and the Schattern-$p$ norm property $\|O\|_1\le \sqrt{\text{rank}(O)}\|O\|_2$, we can bound the Hilbert-Schmidt norm of a Hermitian matrix $\overline{O}$ by $\hsnorm{\overline{O}}\le2\hsnorm{O}.$
Therefore, (\ref{eq:h-moment-convexity-1}) is bounded by 
\begin{align}
\label{eq:h-moment-convexity-2}
\expectDistrOf{v\sim\cD(\rho)}{
\left[\bra{v}O\ket{v}-\expectDistrOf{v\sim\cD(\rho)}{\bra{v}O\ket{v}}\right]^h} 
&\le 
2^{h-1}\left(\frac{(Ch)^h}{d^h} 
+\frac{1}{(d+1)^h} \right) \|O\|_{\text{HS}}^h \nonumber \\
& \le 2^{h-1}\left(\frac{2^h h^h(C^h+1)}{(d+1)^h} \right)  \hsnorm{O}^h \nonumber \\
& \le \frac{(C'h)^h}{(d+1)^h}  \hsnorm{O}^h 
\end{align}
In the last step, we choose constant $C'$ such that $2^{2h-1}(C^h+1)\le (C')^h$.

To complete the proof, we need to further show that the inequality (\ref{eq:h-moment-convexity-2}) holds for odd $h$. In a probability theory setting, Jensen’s inequality for a convex function $f(x)=x^2$ can be written as $(\mathbb{E}{X})^2\le \mathbb{E}{(X^2)}$, where $X$ is a random variable. Let $X=|\bra{v}O\ket{v}-\expectDistrOf{v\sim\cD(\rho)}{\bra{v}O\ket{v}}|^h$. Applying the Jensen’s inequality, we have
\begin{align*}
\expectDistrOf{v\sim\cD(\rho)}{
|\bra{v}O\ket{v}-\expectDistrOf{v\sim\cD(\rho)}{\bra{v}O\ket{v}}|^h} 
&\le \sqrt{\expectDistrOf{v\sim\cD(\rho)}{
|\bra{v}O\ket{v}-\expectDistrOf{v\sim\cD(\rho)}{\bra{v}O\ket{v}}|^{2h}} } \\
& \le \sqrt{\frac{(C'2h)^{2h}}{(d+1)^{2h}}  \hsnorm{O}^{2h}} \\
& = \frac{(C'2h)^h}{(d+1)^h}  \hsnorm{O}^h. 
\end{align*}
In the second line we use the inequality (\ref{eq:h-moment-convexity-2}) by setting $h\to2h$. Therefore, by choosing $C'\to 2C'$, we complete the proof for the odd $h$ case.
\end{proof}
\subsection{Proof of Theorem \ref{thm:result-algorithm}}
\begin{proof}
\algnewcommand\algorithmicinput{\textbf{input:}}
\algnewcommand\Input{\item[\algorithmicinput]}

\algnewcommand\algorithmicgoal{\textbf{goal:}}
\algnewcommand\Goal{\item[\algorithmicgoal]}

We describe the robust classical shadow tomography algorithm in Algorithm~\ref{alg:robust_classical_shadow}.

\begin{algorithm}[h]
\caption{Robust classical shadow tomography}
\label{alg:robust_classical_shadow}
\begin{algorithmic}[1]
\Input The corruption parameter $\gamma$, the input state $\{\rho_i\}_{i=1}^n$ which is $n$ copies of $\rho$ with $\gamma$-corruption ($n$ will be determined later in the analysis), the observables $\{O_j\}_{j=1}^{M}$.
\Goal Return $\{\hat{o}_j\}_{j=1}^{M}$ such that $|\hat{o}_j-\tr(\Obs_j\rho)|\le O(\gamma)$ for all $j\in[M]$. 
\For{$i\in[n]$}
\State Sample a unitary $U_i$ from Haar random unitary.
\State Apply $U_i$ on $\rho_i$: $\rho_i\to \rho_i'=U_i\rho U_i^\dagger$.
\State Measure $\rho_i'$ in the computational basis. Obtain $b_i$.
\State Define $\ket{v_i}:=U^{\dagger}_i\ket{b_i}.$

\State Calculate the classical description of $\hat{\rho}_i:=(d+1)(\ket{v_i}\bra{v_i})-\mathbb{I}_d$.
\EndFor
\For{$j\in[M]$}
    \For{$i\in[n]$}
    \State Compute $o_j^i:=\tr(O_j\hat{\rho}_i)$.
    \EndFor
\State $\hat{o}_j \leftarrow \mathtt{TruncatedMean}(\{o_j^i\}_{i=1}^{n},2\gamma)$. \Comment{$\mathtt{TruncatedMean}(\cdot)$ is defined in Algorithm~\ref{alg:truncated-mean}.} 

\EndFor
\Return $\{\hat{o}_j\}_{j=1}^{M}$
\end{algorithmic}
\end{algorithm}

    Fix $j\in[M]$. Define the random variable
    $$X_j := \frac{(d+1)\langle v|O_j|v\rangle - \Tr[O_j]}{C' h \, \|O_j\|_{HS}},$$
    where $v\sim\mathcal{D}(\rho)$, and denote its mean by $$\mu_j := \mathbb{E}_{v\sim\mathcal{D}(\rho)}[X_j]=\frac{(d+1)\expectDistrOf{v\sim\cD(\rho)}{\bra{v}O_j\ket{v}}-\Tr[O_j]}{C'h \hsnorm{O}}=\frac{\Tr[O_j\rho]}{C'h \hsnorm{O}}.$$
    Lemma \ref{lemma-h-central-moment} implies that for any $h\ge2$ the $h$-th central moment of $X_j$ is bounded,
    $$\expectDistrOf{v\sim\cD(\rho)}{|X_j-\mu_j|^h}\leq 1$$
    According to Lemma \ref{lemma-bounded-central-moment}, for any random variable $X_j'$ whose distribution $\cD_{X_j'}$ is a $\remove$-removing of $\cD(\rho)$, where $0<\remove<1/2$ is a constant, it holds that $\eta_j:=\big|\expect{X_j'}-\expect{X_j}\big|\le O(\remove^{1-\frac{1}{h}})$.
    Choose $\remove=4\corruption$ and $h\ge \log(1/\corruption)$. 
    We have
    $$\big|\expect{X_j'}-\expect{X_j}\big|\le O(\gamma).$$
    Draw $n$ independent samples $v_1,\dots,v_n\stackrel{\mathrm{iid}}{\sim}\mathcal{D}(\rho)$, and 
    define
    $$X_{j,i} := \frac{(d+1)\langle v_i|O_j|v_i\rangle - \Tr[O_j]}{C' h \, \|O_j\|_{HS}}\qquad\text{and}\qquad S_j := \{X_{j,1},\dots,X_{j,n}\}.$$
    We can view the attack in which the adversary  corrupts $\corruption$-fraction of $\rho$ as corrupting $\corruption$-fraction of $S_j$.

    Let $S_j^{\corruption}$ be the $\corruption$-corruption of $S_j$, and $S_j^{\remove/2}$ be the $\remove/2$-truncation of $S_j^{\corruption}$.
    By Lemma~\ref{lemma-robustness-estimation}, if $n \ge O\big(\tfrac{\log(1/\delta')}{(\remove-\gamma)^2}\big) = O\big(\tfrac{\log(1/\delta')}{\gamma^2}\big)$, $S_j^{\remove/2}$ has empirical mean $\mathtt{TruncatedMean}\left(\{X_{j,i}\}_{i=1}^{n}, 2\gamma\right)$ (defined in Algorithm \ref{alg:truncated-mean}) that satisfies
    \begin{equation}
    \label{eq:thm-sample-complexity-1}
        \left|\mathtt{TruncatedMean}\left(\{X_{j,i}\}_{i=1}^{n},2\gamma \right)-\mu_j\right|\le \eta_i=O(\gamma)
    \end{equation}
    with probability greater than $1-\delta'$.

    By multiplying $C' h \, \|O_j\|_{HS}$ on both side of equation (\ref{eq:thm-sample-complexity-1}), we have 
    \begin{equation}
    \label{eq:thm-sample-complexity-2}
        \left|\hat{o}_j
        -\Tr[O_j\rho]\right|\le O\left(h \, \|O_j\|_{HS}\gamma\right)=O\left(\|O_j\|_{HS}\gamma\log(\frac{1}{\gamma})\right),
    \end{equation}
    where $\hat{o}_j:=\mathtt{TruncatedMean}\left(\{(d+1)\langle v_i|O_j|v_i\rangle - \Tr[O_j]\}_{i=1}^{n}, 2\gamma\right)$.

    As a result, for a fixed $j\in[M]$, it holds that
    \[
    \left|\hat{o}_j
        -\Tr[O_j\rho]\right|\le \varepsilon_j
    \]
    with probability greater than $1-\delta'$, where $\varepsilon_j:=\left(\|O_j\|_{HS}\gamma\log(\frac{1}{\gamma})\right)$.
    
    Choosing $\delta'=\delta/M$, by union bound, we have that for all $j\in[M]$, the the estimation $\hat{o}_j$ has an additive error smaller than $\varepsilon:=\max_j\varepsilon_j=O(\max_j\|\Obs_j\|_{HS}\gamma\log(1/\gamma))$ with probability greater than $1-\delta$, and the sample complexity is $n=O(\frac{1}{\gamma^2}\log(\frac{M}{\delta}))=O\left(\frac{\max_j\|\Obs_j\|_{HS}^2}{\epsilon^2}\log^2{\frac{1}{\eps}}\log(\frac{M}{\delta})\right)$. 
\end{proof}

\begin{definition}[Quantile]
\label{def:quantile}
Let $X$ be a random variable and $q\in[0,1]$. 
We say $t$ is the $q\operatorname{-quantile}$ of $X$ if $t$ is the infimum over $t\in R$ such that $\Pr[X\le t]\ge q$.
Let $S$ be a multiset. 
We say $t$ is the $q\operatorname{-quantile}$ of $S$ if $t$ is the $q\operatorname{-quantile}$ of the uniform distribution over $S$.
\end{definition}

\begin{definition}[Remove mass of a distribution\yc{wording}]
    \label{def:removing}
    Let $\mathcal{D}$ be a distribution over $\mathbb{R}$. 
    Let $W\subseteq \mathbb{R}$ be a subset of real number satisfying $\Pr_{x\sim \mathcal{D}}[x\in W]=\remove$.
    Let $D'$ be a distribution over real number that is obtained by sampling a value from $D$ conditioned on the value is not in $W$, i.e.,
    \begin{equation}
        \label{eq:remove_mass}
        \mathcal{D}'(x):=\left\{
        \begin{array}{ll}
            \frac{1}{1-\remove}\mathcal{D}(x),&\:\operatorname{if} x\notin W;\\
            0, &\:\operatorname{if} x\notin W.
        \end{array}
        \right.
    \end{equation}
    We say $\mathcal{D}'$ is $\remove$-removing from $\mathcal{D}$. \yc{wording.}
    \ynote{You can refer to [DK'23] to check their terminology. It may be a good idea to be consistent with them}
\end{definition}

\begin{definition}[Truncation of a distribution \yc{Do we need this?}]
    \label{def:truncation_of_dist}
    Let $D$ be a distribution over real number.
    Let $q_1$ and $q_2$ be the $\truncation$-quantile and $1-\truncation$-quantile respectively.
    Let $D'$ be a distribution over real number obtained by sampling from $D$ conditioned on the value is in $(q_1,q_2)$, i.e.,
    \begin{equation}
        \label{eq:truncation}
        D'(x):=\left\{
        \begin{array}{ll}
            \frac{1}{1-2\truncation}D(x),&\:\operatorname{if} x\in (q_1,q_2);\\
            0, &\:\operatorname{otherwise}.
        \end{array}
        \right.
    \end{equation}
    We say $D'$ is the $\truncation$-truncation of $D$.
\end{definition}

\begin{definition}[Truncation of a multiset.]
    \label{def:truncation_of_dist}
    Let $S$ be a multiset over real number.
    Let $q_1$ and $q_2$ be the $\truncation$-quantile and $1-\truncation$-quantile of $S$ respectively.
    Define $S':=\{x:x\in S\cap (q_1,q_2)\}$.
    We say $S'$ is the $\truncation$-truncation of $S$.
\end{definition}

\begin{lemma}[Robustness estimation \yc{wording} (Proposition 1.18 of \cite{diakonikolas2023algorithmic})]
\label{lemma-robustness-estimation}
    Let $0<\corruption<\remove<1/2$.
    Let $\cD$ be a distribution with mean $\mu$ satisfying that for any $\cD'$ which is a $2\remove$-removing of $\cD$, it holds that $\big|\expectDistrOf{x\sim\cD}{x}-\expectDistrOf{x\sim\cD'
    }{x}\big|\le\eta$.
    Let $S_0$ be a set of $n$ independent samples from $\cD$ and $S$ be a $\corruption$-corruption of $S_0$.
    Let $S_\truncation$ be the $\truncation$-truncation of $S$, and let $\hat{\mu}$ be the mean of $S_\truncation$.
    If $n\ge \frac{\log(1/\delta)}{(\remove-\corruption)^2}$, then $|\hat{\mu}-\mu|\le\eta$ with probability greater than $1-\delta$.
\end{lemma}

\begin{lemma}[Robustness of bounded $h$th central moment. \yc{wording}]
\label{lemma-bounded-central-moment}
    Let $X$ be a random variable with distribution $\cD_{X}$ and mean $\mu$.
    Let $h\in\mathbb{N}$.
    If $\expect{|X-\mu|^h}\le 1$, then for any random variable $X'$ whose distribution $\cD_{X'}$ is a $\remove$-removing of $\cD$ where $\remove<1/2$, it holds that $\big|\expect{X'}-\expect{X}\big|\le O(\remove^{1-\frac{1}{h}})$.
\end{lemma}
\begin{proof}
    Without loss of generality, we assume $\expect{X'}-\expect{X}\ge 0$.
    Because $\cD_{X'}$ is a $\xi$-removing of $\cD_{X}$, we have $\big|\Pr[X'-\mu\ge t]-\Pr[X-\mu\ge t]\big|\le \remove$ for all $t\in\mathbb{R}$. Hence, $\Pr[X'-\mu\ge t]\le \Pr[X-\mu\ge t]+\remove$.
    
    Besides, it holds that $\Pr[X'-\mu\ge t] \le \frac{1}{1-\remove}\Pr[X-\mu\ge t]$ for all $t\in \mathbb{R}$.
    We have $\Pr[X'-\mu\ge t]\le \Pr[X-\mu\ge t] +\frac{\remove}{1-\remove} \Pr[X-\mu\ge t] =\Pr[X-\mu\ge t] + \Pr[X-\mu\ge t]$ for $\remove<1/2$.
    By using the inequality for the $h$th moment $Pr[|X-\mu|\ge t]\le \frac{\expect{|X-\mu}}{t^h}$,
    when $\expect{|X-\mu|^h}\le 1$, we have  $\Pr[|X'-\mu|\ge t]\le \Pr[|X-\mu|\ge t] + 1/t^h$.

    Therefore, $\Pr[X'-\mu\ge t]\le \Pr[|X-\mu|\ge t]+\min(\remove, 1/t^h)$. 

    Similarly, we have $\Pr[\mu-X'\ge t]\ge \Pr[\mu-X \ge t] - \remove$.
    Also, the removal of the mass cannot be greater than the original mass.
    It holds that $\Pr[\mu-X'\ge t]\ge \Pr[\mu-X\ge t] - 1/t^h$.

    Therefore, $\Pr[\mu-X'\ge t]\ge \Pr[\mu-X\ge t] -\max(\remove, 1/t^h)$.
    
    Now we are ready to upper bound the difference between $\expect{X'}$ and $\expect{X'}$.
    Let $p_+:=\Pr[X'\ge \mu]$ and $p_-:=\Pr[X'\le \mu]$. 
    We have 
    \begin{align*}
        \expect{X'}-\expect{X} 
        &=   \expect{X'}-\mu\\
        &=\expect{X'-\mu}\\
        &= p_{+}\expectCond{X'-\mu
        }{X'-\mu\ge 0} + p_{-} \expectCond{X'-\mu
        }{X' - \mu\ge 0}\\
        &=p_{+}\int_0^\infty \Pr[X'-\mu\ge t] \dm t- p_{-}\int_0^\infty\Pr[\mu-X'\ge t] \dm t\\
        &\le p_{+}\int_0^\infty \Pr[X-\mu\ge t] + \min\{\remove,1/t^h\} \dm t \\&\quad -p_{-}\int_0^\infty\Pr[\mu-X'\ge t]-\max\{\remove,1/t^h\} \dm t\\
        &=\int_0^\infty \min\{\remove,1/t^h\} \dm t\\
        & = \int_0^{\remove^{-1/h}} \remove \dm t + \int_{\remove^{-1/h}}^\infty 1/t^h \dm t\\
        &=(\frac{h}{h-1})\remove^{1-1/h}.
    \end{align*}
    
\end{proof}

\section{Quantum state tomography}

In this section, we show how to utilize the robust shadow tomography that we have designed in Section \ref{sec:shadow-tomograghy} to create a robust quantum state tomography. At a high level, we can reduce the problem of learning a quantum state to the problem of finding (one) of the closest quantum states to the unknown state in a covering set of quantum states. 
Let us start by constructing a covering set of quantum states of a given rank. 
\begin{lemma}\label{lem:covering}
    Let $1\le r \le d$ be integers and $\eps>0$. There exists a set $\mathcal{N}$ of quantum states of rank at most $r$ such that:
    \begin{itemize}
        \item $ |\mathcal{N}| \le \left(1+ \frac{4r}{\eps}\right)^r \left( \frac{8}{\eps}\right)^{rd}$,
        \item for every quantum state $\rho$ of rank $r$, there exists a quantum state $\sigma \in \mathcal{N}$ such that $\|\rho - \sigma \|_{\tr}\le \eps$. 
    \end{itemize}
\end{lemma}
\begin{proof}
The set of pure states $\mathcal{S}_{d,1}$ can be identified by the set $S^{d}$ of unit vectors in $\mathbb{C}^d$. 
    Since we have that $N(S^{d}, \|\cdot\|_2, \eps/4)\le (8/\eps)^d$ \cite[Lemma 5.3]{aubrun2017alice}, we denote by $\mathcal{N}_0$ a set of pure states achieving this bound. Then,  we can construct the following set:
    \begin{align}
        \mathcal{N} = \left\{ \rho = \frac{1}{\sum_{i=1}^r \lambda_i}\sum_{i=1}^r \lambda_i \proj{\phi_i} : \phi_i\in \mathcal{N}_0, \lambda_i \in \frac{\eps}{4r}\mathbb{Z}\cap [0,1] , 1-\frac{\eps}{4}\le \sum_{i=1}^r \lambda_i \le 1\right\}
    \end{align}
    We have that $|\mathcal{N}| \le \left(1+ \frac{4r}{\eps}\right)^r |\mathcal{N}_0|^r = \left(1+ \frac{4r}{\eps}\right)^r \left( \frac{8}{\eps}\right)^{rd}$. 

We claim that $\mathcal{N}$ is a covering set for the set of quantum states of rank $r$. To prove this, we take $\sigma$ a quantum state of rank $r$, it can be written as:
\begin{equation}
    \sigma = \sum_{i=1}^r \mu_i\proj{\psi_i}.
\end{equation}
Let $\lambda_i \in \frac{\eps}{4r}\mathbb{Z}\cap [0,1] $ such that $0\le \mu_i-\lambda_i \le \frac{\eps}{4r}$. This gives $1-\frac{\eps}{4}\le \sum_{i=1}^r \lambda_i \le 1$. Similarly, let $\phi_i\in \mathcal{N}_0$ such that $\|\phi_i - \psi_i\|_2\le \frac{\eps}{4}$. Construct $\rho = \frac{1}{\sum_{i=1}^r \lambda_i} \sum_{i=1}^r \lambda_i \proj{\phi_i} \in \mathcal{N}$. We have that 
\begin{align}
    \|\rho - \sigma \|_{\tr} &= \left\|  \frac{1}{\sum_{i=1}^r \lambda_i}\sum_{i=1}^r \lambda_i \proj{\phi_i}- \sum_{i=1}^r \mu_i\proj{\psi_i} \right\|_{\tr}
    \\&\le  \left\| \sum_{i=1}^r \lambda_i \proj{\phi_i}- \sum_{i=1}^r \mu_i\proj{\psi_i} \right\|_{\tr} + \left\|  \frac{1}{\sum_{i=1}^r \lambda_i}\sum_{i=1}^r \lambda_i \proj{\phi_i} - \sum_{i=1}^r \lambda_i \proj{\phi_i} \right\|_{\tr}
    \\&\le \left\|  \sum_{i=1}^r \lambda_i \proj{\phi_i}- \sum_{i=1}^r \mu_i\proj{\phi_i} \right\|_{\tr} + \left\|  \sum_{i=1}^r  \mu_i \proj{\phi_i}- \sum_{i=1}^r  \mu_i\proj{\psi_i} \right\|_{\tr} + \frac{1}{2} \left|1 -\sum_{i=1}^r \lambda_i  \right|
    \\&\le \sum_{i=1}^r |\lambda_i - \mu_i| + \sum_{i=1}^r \mu_i   \left\|  \proj{\phi_i}- \proj{\psi_i} \right\|_{\tr} +\frac{\eps}{8}
    \\&\le \sum_{i=1}^r |\lambda_i - \mu_i| + 2\sum_{i=1}^r \mu_i   \left\| \phi_i- \psi_i \right\|_{2}+\frac{\eps}{8}
    \\&\le  \sum_{i=1}^r \frac{\eps}{4r} + 2\sum_{i=1}^r \mu_i   \frac{\eps}{4} +\frac{\eps}{8} < \eps. 
\end{align}

\end{proof}
Combining this covering result with the robust shadow tomography we can propose a robust state tomography algorithm. 
\begin{theorem}
    Under $\gamma$-adversarial corruption, there exists an full state tomography algorithm that achieves an error of $\eps = O(\gamma\sqrt{\rk})$ in trace distance using $\ns = \tildeO{\dims\rk^2/\eps^2}= \tildeO{\dims\rk/\gamma^2}$ samples.
\end{theorem}
\begin{proof}
    We will use the covering construct in Lemma \ref{lem:covering}. Namely, we have an $\eps/5$-covering set $\mathcal{N}$ of size $M=|\mathcal{N}| \le \left(1+ \frac{20r}{\eps}\right)^r \left( \frac{40}{\eps}\right)^{rd}$. We list the elements of this set as $\cN = \{\rho_i\}_{i\in [M]}$. For each $i \neq j$, we construct the Holevo–Helstrom observable $ O_{i,j} $ satisfying:
    \begin{equation}
        \tr(O_{i,j}(\rho_i-\rho_j)) = \|\rho_i-\rho_j\|_{\tr}.
    \end{equation}
    Since $\rho_i$ has rank at most $r$, $O_{i,j}$ has rank at most $2r$ thus 
    $\hsnorm{O_{i,j}}\le \sqrt{2r}$. 
    
    Using our robust classical shadows of Theorem \ref{thm:result-algorithm}, we can approximate $(\tr(O_{i,j}\rho))_{i,j \in [M]}$ simultaneously to within an error $\eps$ using a number of copies 
    \begin{equation}
         \ns = \bigO{\frac{\max_{i\in[\nobs]}\hsnorm{O_i}^2}{\eps^2}\log\frac{\nobs^2}{\delta}}  = \bigO{\frac{r^2d\log(40/\eps)+r^2\log(1+20r/\eps)+r\log(1/\delta)}{\eps^2}}.
    \end{equation}
    Concretly, we obtain $(E_{i,j})_{i,j}$ such that for all $i,j \in [M]$, we have $|\tr(O_{i,j}\rho)-E_{i,j}|\le \eps/5$. 
    Our approximation of the unknown quantum state is then given by \cite{Yatracos1985Jun,buadescu2021improved},
    \begin{equation}\label{eq:stat}
        \hat{\rho} = \rho_l  \quad \text{ where } \quad   l\in \argmin_{l \in [M]}\left\{ \max_{i\in [M]} \left|\tr(O_{i,l}\rho_l) -E_{i,l}\right| \right\}.
    \end{equation}
    Since $\mathcal{N}$ is a covering, there exists $\rho_k \in \mathcal{N}$ such that $\|\rho - \rho_k\|_{\tr}\le \eps/5$. We have that from \eqref{eq:stat}:
    \begin{equation}
        \left|\tr(O_{k,l}\rho_l) -E_{k,l}\right| \le  \max_{i}  \left|\tr(O_{i,k}\rho_k) -E_{i,k}\right| \le  \|\rho - \rho_k\|_{\tr}+\eps/5\le 2\eps/5. 
    \end{equation}
   We deduce that by the triangle inequality:
    \begin{align}
        \|\hat{\rho} - \rho\|_{\tr} &\le    \|\rho_l - \rho_k\|_{\tr} +\|\rho_k - \rho\|_{\tr}
        \\&\le  \tr(O_{k,l}( \rho_k-\rho_l)) +\eps/5
        \\&=   E_{k,l} -\tr(O_{k,l} \rho_l)+ \tr(O_{k,l} \rho_k) -E_{k,l}     +\eps/5
        \\&\le 2\eps/5 +\tr(O_{k,l}\rho) -E_{k,l}+ 2\eps/5 
        \\&\le \eps. 
    \end{align}
\end{proof}

\section*{Acknowledgement}
Vladimir Braverman is supported by NSF 2528780 and the Naval Research (ONR) grant N00014-23-1-2737. Nai-Hui Chia is supported by NSF Award FET-2243659, NSF Career Award FET-2339116, and DOE Quantum Testbed Finder Award DE-SC0024301. Chia-Ying Lin is supported by DOE Quantum Testbed Finder Award DE-SC0024301 and NSF Award FET-2243659. Yu-Ching Shen is supported by NSF Award FET-2243659 and NSF Career Award FET-2339116. Yuhan Liu is supported by the Naval Research (ONR) grant N00014-23-1-2737.

\bibliography{refs}
\bibliographystyle{alpha}

\end{document}